\documentclass[copyright,creativecommons]{eptcs}

\usepackage{amsthm}
\usepackage[english]{babel}
\usepackage{amsmath}
\usepackage{amssymb}
\usepackage{array}
\usepackage{stmaryrd}
\usepackage{graphicx}
\usepackage{xspace}
\usepackage{calc}

\hypersetup{colorlinks,urlcolor=black} 
\usepackage{mdwlist}

\usepackage{float}
\floatstyle{boxed} \restylefloat{figure}

\usepackage{type1cm}
\makeatother

\newdimen\proofrulebreadth \proofrulebreadth=.05em
\newdimen\proofdotseparation \proofdotseparation=1.25ex
\newdimen\proofrulebaseline \proofrulebaseline=2ex
\newdimen\proofrulebaseline \proofrulebaseline=1.7ex
\newcount\proofdotnumber \proofdotnumber=3
\let\then\relax
\def\hfi{\hskip0pt plus.0001fil}
\mathchardef\squigto="3A3B
\newif\ifinsideprooftree\insideprooftreefalse
\newif\ifonleftofproofrule\onleftofproofrulefalse
\newif\ifproofdots\proofdotsfalse
\newif\ifdoubleproof\doubleprooffalse
\let\wereinproofbit\relax
\newdimen\shortenproofleft
\newdimen\shortenproofright
\newdimen\proofbelowshift
\newbox\proofabove
\newbox\proofbelow
\newbox\proofrulename
\def\shiftproofbelow{\let\next\relax\afterassignment\setshiftproofbelow\dimen0 }
\def\shiftproofbelowneg{\def\next{\multiply\dimen0 by-1 }%
\afterassignment\setshiftproofbelow\dimen0 }
\def\setshiftproofbelow{\next\proofbelowshift=\dimen0 }
\def\setproofrulebreadth{\proofrulebreadth}

\def\prooftree{
\ifnum  \lastpenalty=1
\then   \unpenalty
\else   \onleftofproofrulefalse
\fi
\ifonleftofproofrule
\else   \ifinsideprooftree
        \then   \hskip.5em plus1fil
        \fi
\fi
\bgroup
\setbox\proofbelow=\hbox{}\setbox\proofrulename=\hbox{}%
\let\justifies\proofover\let\leadsto\proofoverdots\let\Justifies\proofoverdbl
\let\using\proofusing\let\[\prooftree
\ifinsideprooftree\let\]\endprooftree\fi
\proofdotsfalse\doubleprooffalse
\let\thickness\setproofrulebreadth
\let\shiftright\shiftproofbelow \let\shift\shiftproofbelow
\let\shiftleft\shiftproofbelowneg
\let\ifwasinsideprooftree\ifinsideprooftree
\insideprooftreetrue
\setbox\proofabove=\hbox\bgroup$\displaystyle 
\let\wereinproofbit\prooftree
\shortenproofleft=0pt \shortenproofright=0pt \proofbelowshift=0pt
\onleftofproofruletrue\penalty1
}

\def\eproofbit{
\ifx    \wereinproofbit\prooftree
\then   \ifcase \lastpenalty
        \then   \shortenproofright=0pt  
        \or     \unpenalty\hfil         
        \or     \unpenalty\unskip       
        \else   \shortenproofright=0pt  
        \fi
\fi
\global\dimen0=\shortenproofleft
\global\dimen1=\shortenproofright
\global\dimen2=\proofrulebreadth
\global\dimen3=\proofbelowshift
\global\dimen4=\proofdotseparation
\global\count255=\proofdotnumber
$\egroup  
\shortenproofleft=\dimen0
\shortenproofright=\dimen1
\proofrulebreadth=\dimen2
\proofbelowshift=\dimen3
\proofdotseparation=\dimen4
\proofdotnumber=\count255
}

\def\proofover{
\eproofbit 
\setbox\proofbelow=\hbox\bgroup 
\let\wereinproofbit\proofover
$\displaystyle
}%
\def\proofoverdbl{
\eproofbit 
\doubleprooftrue
\setbox\proofbelow=\hbox\bgroup 
\let\wereinproofbit\proofoverdbl
$\displaystyle
}%
\def\proofoverdots{
\eproofbit 
\proofdotstrue
\setbox\proofbelow=\hbox\bgroup 
\let\wereinproofbit\proofoverdots
$\displaystyle
}%
\def\proofusing{
\eproofbit 
\setbox\proofrulename=\hbox\bgroup 
\let\wereinproofbit\proofusing
\kern0.3em$
}

\def\endprooftree{
\eproofbit 
  \dimen5 =0pt
\dimen0=\wd\proofabove \advance\dimen0-\shortenproofleft
\advance\dimen0-\shortenproofright
\dimen1=.5\dimen0 \advance\dimen1-.5\wd\proofbelow
\dimen4=\dimen1
\advance\dimen1\proofbelowshift \advance\dimen4-\proofbelowshift
\ifdim  \dimen1<0pt
\then   \advance\shortenproofleft\dimen1
        \advance\dimen0-\dimen1
        \dimen1=0pt
        \ifdim  \shortenproofleft<0pt
        \then   \setbox\proofabove=\hbox{%
                        \kern-\shortenproofleft\unhbox\proofabove}%
                \shortenproofleft=0pt
        \fi
\fi
\ifdim  \dimen4<0pt
\then   \advance\shortenproofright\dimen4
        \advance\dimen0-\dimen4
        \dimen4=0pt
\fi
\ifdim  \shortenproofright<\wd\proofrulename
\then   \shortenproofright=\wd\proofrulename
\fi
\dimen2=\shortenproofleft \advance\dimen2 by\dimen1
\dimen3=\shortenproofright\advance\dimen3 by\dimen4
\ifproofdots
\then
        \dimen6=\shortenproofleft \advance\dimen6 .5\dimen0
        \setbox1=\vbox to\proofdotseparation{\vss\hbox{$\cdot$}\vss}%
        \setbox0=\hbox{%
                \advance\dimen6-.5\wd1
                \kern\dimen6
                $\vcenter to\proofdotnumber\proofdotseparation
                        {\leaders\box1\vfill}$%
                \unhbox\proofrulename}%
\else   \dimen6=\fontdimen22\the\textfont2 
        \dimen7=\dimen6
        \advance\dimen6by.5\proofrulebreadth
        \advance\dimen7by-.5\proofrulebreadth
        \setbox0=\hbox{%
                \kern\shortenproofleft
                \ifdoubleproof
                \then   \hbox to\dimen0{%
                        $\mathsurround0pt\mathord=\mkern-6mu%
                        \cleaders\hbox{$\mkern-2mu=\mkern-2mu$}\hfill
                        \mkern-6mu\mathord=$}%
                \else   \vrule height\dimen6 depth-\dimen7 width\dimen0
                \fi
                \unhbox\proofrulename}%
        \ht0=\dimen6 \dp0=-\dimen7
\fi
\let\doll\relax
\ifwasinsideprooftree
\then   \let\VBOX\vbox
\else   \ifmmode\else$\let\doll=$\fi
        \let\VBOX\vcenter
\fi
\VBOX   {\baselineskip\proofrulebaseline \lineskip.2ex
        \expandafter\lineskiplimit\ifproofdots0ex\else-0.6ex\fi
        \hbox   spread\dimen5   {\hfi\unhbox\proofabove\hfi}%
        \hbox{\box0}%
        \hbox   {\kern\dimen2 \box\proofbelow}}\doll%
\global\dimen2=\dimen2
\global\dimen3=\dimen3
\egroup 
\ifonleftofproofrule
\then   \shortenproofleft=\dimen2
\fi
\shortenproofright=\dimen3
\onleftofproofrulefalse
\ifinsideprooftree
\then   \hskip.5em plus 1fil \penalty2
\fi
}

\newcommand\nw[1]{#1^{\scalebox{.4}{\NEW}}}
\newcommand\lam[1]{\lambda #1.}
\def\apart{{\#}}

\newcommand\app{\mathtt{app}}
\newcommand\vr{\mathtt{var}}
\newcommand\fn{\mathtt{lam}}

\newcommand\qleq{\mathrel{{}_{\scriptstyle {?}}{\approx}}}

\newcolumntype{L}{>{$}l<{$}}
\newcolumntype{C}{>{$}c<{$}}
\newcolumntype{R}{>{$}r<{$}}

\newcounter{edefinitioncount}

{\par\addtolength{\baselineskip}{.8mm}\begin{tabbing}\hspace{5mm}\=\hspace{5mm}\=\hspace{5mm}\=\kill}
{\end{tabbing}}

\newtheoremstyle{jamiestyle}
 {4pt}
 {0pt}
 {\it}
 {0pt}
 {\bf}
 {.}
 { }
 {}
\theoremstyle{jamiestyle}
\newtheorem{thrm}{Theorem}[section]
\newtheorem{prop}[thrm]{Proposition}
\newtheorem{lemm}[thrm]{Lemma}
\newtheorem{corr}[thrm]{Corollary}
\newtheoremstyle{jamienfstyle}
  {4pt}
  {0pt}
  {\normalfont}
  {0pt}
  {\bf}
  {.}
  { }
  {}
\theoremstyle{jamienfstyle}

\newtheorem{defn}[thrm]{Definition}
\newtheorem{xmpl}[thrm]{Example}
\newtheorem{rmrk}[thrm]{Remark}

\newcommand\id{\f{id}}

\newcommand{\rewriteswith}[1]{\stackrel{#1}{\to}}
\newcommand{\rewriteswithc}[1]{\stackrel{#1}{\to_c}}
\newcommand{\tworewriteswith}[1]{\stackrel{#1}{\leftrightarrow}}
\newcommand{\tworewriteswithc}[1]{\stackrel{#1}{\leftrightarrow_c}}

\newcommand\aleq{\mathrel{{\approx}_{\scriptstyle {\alpha}}}}

\newcommand\new{\reflectbox{$\mathsf N$}}

\newcommand{\f}[1]{\ensuremath{\mathit{#1}}}

\newcommand\act{\cdot}

\newcommand{\liff}{\ensuremath{\Leftrightarrow}}

\newlength{\insidewd}

\newlength{\stackht}
\newcommand{\deffont}[1]{\textbf{#1}}

\newcommand{\rulefont}[1]{\ensuremath{(\mathbf{#1})}}

\newcommand{\Id}{\mathit{id}}

\newcommand{\cent}{\vdash}

\newcommand{\centt}[1]{\cent_{_{\theory{#1}}}}

\newcommand{\ssm}{\mapsto}
\newcommand{\sm}{\mapsto}

\newcommand{\theory}[1]{\ensuremath{\mathsf{#1}}}

\newcommand{\tf}[1]{\mathsf{#1}}

\newcommand\NEW[0]{\reflectbox{\ensuremath{\mathsf{N}}}}

\def\mone{{\text{-}1}}

\title{Closed nominal rewriting and efficiently computable nominal
  algebra equality}
\author{\href{http://www.dcs.kcl.ac.uk/staff/maribel}{Maribel Fern\'andez} \institute{\href{http://www.dcs.kcl.ac.uk/staff/maribel}{www.dcs.kcl.ac.uk/staff/maribel}} \and \href{http://www.gabbay.org.uk}{Murdoch J. Gabbay} \institute{\href{http://www.gabbay.org.uk}{www.gabbay.org.uk}}}
\def\titlerunning{Efficiently computable nominal equality}
\def\authorrunning{M.~Fern\'andez and M.~Gabbay}
\begin{document}
\maketitle

\begin{abstract}
  We analyse the relationship between nominal algebra and nominal
  rewriting, giving a new and concise presentation of equational
  deduction in nominal theories.  With some new results, we
  characterise a subclass of equational theories for which nominal
  rewriting provides a complete  procedure to check 
  nominal algebra equality.  This subclass includes specifications of
  lambda-calculus and first-order logic. 
\end{abstract}

\section{Introduction}
\label{sec.introduction}

It is very common, when formally defining a programming language,
computation model, or deduction system, to appeal to operators with
binding like $\forall$, $\lambda$, $\nu$, or $\int$.  We are therefore
interested in frameworks with support for the specification, analysis
and evaluation of operators with binding mechanisms. Such frameworks
are needed not only in logic and theoretical computer science (where
binders like $\forall$, $\lambda$ and $\nu$ are familiar) but also to
mechanise mathematics, linguistics, systems biology, and so on.

First, we need to define the notion of a binder.  One answer is to
identify all binding with functional abstraction $\lambda$.  This
approach is taken in the definition of higher-order abstract
syntax~\cite{pfenning:hoas,pfenning:prirh}, higher-order
rewriting~\cite{NipkowT:higors}, and combinatory reduction
systems~\cite{KlopJW:comrsis}, amongst others.  Since higher-order
unification is undecidable, and it plays a key role in rewriting
(e.g., rewrite steps are computed using matching, critical pairs are
computed using unification),
most higher-order rewrite formalisms only use \emph{higher-order
  patterns}~\cite{MillerD:unistl}, a decidable sublanguage.  This
fact already suggests that names and binding might be a simpler
concept, and computationally more tractable, than raw functional
abstraction.

In fact, it has been shown that higher-order patterns correspond
almost exactly to \emph{nominal
  terms}~\cite{LevyJ:rta08,gabbay:perntu-jv}.  This correspondence is
robust, and extends to solutions of unification problems
\cite{gabbay:perntu-jv}, and also in the presence of arbitrary
equality theories \cite{gabbay:unialt}.  Unification and matching of
nominal terms are decidable~\cite{gabbay:nomu-jv} and efficient
(see~\cite{FernandezM:wollic-jv} for a linear-time nominal matching
algorithm, and~\cite{CalvesC:phd,LevyJ:effnua} for efficient unification
algorithms).  Nominal terms have been the basis of rewriting
\cite{gabbay:nomr-jv}, logic programming \cite{cheney:nomlp}, and
algebra \cite{gabbay:nomuae}.

Nominal terms are like first-order terms (`standard' syntax) but come
supplied with \emph{atoms}; a kind of bindable constant with semantics
discussed in \cite{gabbay:newaas-jv}.  Atoms display special behaviour
which will be developed in the body of the paper.  For now, we
illustrate the use of nominal terms to express a theory of
$\beta\eta$-equivalence in nominal algebra.

Suppose term-formers $\tf{lam}:1$ (the number indicates arity) and $\tf{app}:2$.
Then:
$$
\begin{array}{l@{\ }l@{\ }l}
\tf{lam}([a]\tf{lam}([b]\tf{app}(a,b))) & \text{represents the $\lambda$-term} & \lam{f}\lam{x}fx
\\
\tf{lam}([a]\tf{lam}([b]X)) & \text{represents a $\lambda$-term schema} &\lam{x}\lam{y}t
\end{array}
$$
Here $a$ and $b$ are atoms, $[a]\text{-}$ is \emph{atoms-abstraction}, with special properties we discuss later, and 
$X$ is a variable which corresponds to \emph{meta-variables} like $t$ above.
But $X$ is not a meta-variable; it is a variable in nominal terms.
To avoid confusion, we call variables in nominal terms \emph{unknowns}.

We can define $\beta$-reduction and $\eta$-reduction as follows (see~\cite{gabbay:nomr} for an alternative nominal rewriting system that uses an explicit substitution operator):
\[\begin{array}{c@{\quad}llcl}
\rulefont{\beta_{\app}} & &\tf{app}(\tf{lam}([a]\tf{app}(X,X')),Y)& \to  & 
\\
&&&&\hspace{-5em}\tf{app}(\tf{app}(\tf{lam}([a]X',Y),\ \tf{app}(\tf{lam}([a]X),Y)))
\\
\rulefont{\beta_{\vr}} & &\tf{app}(\tf{lam}([a]a),X)&\to& X
\\
\rulefont{\beta_\epsilon} & a\apart Y\cent &\tf{app}(\tf{lam}([a]Y),X)& \to & Y
\\
\rulefont{\beta_\fn} & b\apart Y \cent &\tf{app}(\tf{lam}([a]\tf{lam}([b]X)),Y)&  \to  & \tf{lam}([b]\tf{app}(\tf{lam}([a]X),Y))
\\
\rulefont{\eta} & a\# X \cent & \tf{lam}([a]\tf{app}(X,a)) &\to& X
\end{array}
\]
We obtain a nominal algebra theory just by replacing $\to$ with $=$.

Setting aside the verbosity of the syntax above, what we would like the reader to take from this example is \emph{how close} the specification is to what we write in mathematical discourse. A \emph{freshness side-condition} $a\#X$ formalises the
usual condition  $x\not\in\f{fv}(u)$, using an atom $a$ for $x$ and 
an unknown $X$ for the metavariable $u$.

This motivates nominal algebra~\cite{gabbay:noma-nwpt,gabbay:nomuae}
and also nominal rewriting~\cite{gabbay:nomr,gabbay:nomr-jv}; theories of
equality and rewriting respectively for nominal terms (see also~\cite{CloustonR:nomel}, though this does not use nominal terms). 
The resulting theories have semantics in (nominal) sets and good computational properties; these are investigated in several other papers by the authors and others.
 
The relationship between equational reasoning and rewriting is well understood
in the first-order case where terms do not include binders:  
If an equational theory $E$ can be presented by a terminating
and confluent rewrite system then equality modulo $E$ is
decidable~\cite{DershowitzN:rews,NipkowT:terraa}. Even if the 
rewrite system is not confluent it may still be possible to use rewriting if the
system can be completed by adding new rules~\cite{KnuthD:completion};  
implementations of
equational logic have been based on these observations~\cite{BoyerR:ACL2,ODonnelM:eqlp,GoguenJ:2OBJ,McCuneW:EQP,McCuneW:Otter}.

However, in systems with binding the situation is
different. Semi-automatic tools exist, many relying on higher-order
formalisms that use the $\lambda$-calculus as meta-language, but
since higher-order unification is undecidable in general, higher-order rewriting
frameworks need to restrict the form of the rules to achieve a
decidable rewriting relation. This makes it difficult to define
completion procedures for higher-order rewriting systems.
For nominal systems, the relationship between rewriting and equality
is not straightforward and has not been established yet. This paper
fills this gap. 

The main contributions of this paper are: 
\begin{itemize*}
\item We give new presentations of nominal rewriting and nominal
  algebra that are significantly more concise than those
  in~\cite{gabbay:nomr-jv,gabbay:nomuae}. This gives a clear and
  `user-friendly' overview of the two systems.

\item 
  We identify a completeness result
  (Theorem~\ref{thrm.compl}) which shows  a precise 
  connection between nominal rewriting and nominal algebra.
In other words, we fill the gap mentioned above.

\item We identify a stronger completeness result for a subset of
  nominal rewriting already investigated for its good computational
  properties \cite{gabbay:nomr-jv}: \emph{closed} rewriting.  Closed
  rewriting is sound and complete for nominal algebra
  (Theorem~\ref{thm:sound-compl}), in a particularly direct manner.

\end{itemize*} 
Note that the collection of closed nominal terms is at least as
expressive as other systems in the literature, including Combinatory
Reduction Systems~\cite{KlopJW:comrsis} and Higher-Order Rewriting
Systems~\cite{NipkowT:higors}.  This is discussed in
\cite{gabbay:nomr-jv}.  However, nominal rewrite/algebra systems exist
that do \emph{not} fall into the closed collection.  For instance, the
natural specification of $\pi$-calculus \rulefont{Open} labelled
transition \cite{Milner:calmpII} displays a \emph{gensym}-like
behaviour that, as it happens, is not captured by closed nominal terms
(but can be defined using nominal rewriting rules): $P
\stackrel{a\overline{b}}{\to} Q$ implies $\nu b.P
\stackrel{a\overline{b}}{\to} Q$.\footnote{This rule can be fit
  into the nominal algebra/rewriting framework, e.g. with a bit of
  sugar as follows: $(Z,\ P \stackrel{a\overline{b}}{\to} Q)\to (Z,\ P
  \stackrel{a\overline{b}}{\to} Q,\ \nu b.P
  \stackrel{a\overline{b}}{\to} Q)$.  We are interested in
  expressivity, not elegance, at this point.}

So both our completeness results are relevant.
We cannot say one is `right' and the other `wrong'; nominal terms are more expressive but fewer things are true of them relative to closed nominal terms.  Both have good theorems relating rewriting with equational reasoning, which we describe in this paper. 
 
The rest of the paper is organised as follows: In
Section~\ref{sec.syntax} we recall the basic notions of nominal
syntax.  Section~\ref{ssec.na.theories} gives a new and uniform
presentation of nominal algebra and nominal
rewriting. Section~\ref{sec:compl} compares nominal algebra and
rewriting and establishes a first completeness
result. Section~\ref{sec:closedrewriting} discusses closed nominal
rewriting as an efficient mechanism to implement deduction in nominal
theories, and establishes the soundess and completeness of nominal
rewriting for equational deduction in theories presented by closed
rules. Using this result, we give an algorithm to implement nominal
algebra in an efficient way. We conclude the paper in
Section~\ref{sec:Conclusions}.

\section{Syntax and $\alpha$-equivalence}
\label{sec.syntax}

\emph{Nominal terms} were introduced in \cite{gabbay:nomu-jv}
as a formal syntax for the specification of systems with binding.
 In this section we recall the main notions of nominal syntax; for
 more details and examples we refer the reader
to~\cite{gabbay:nomu-jv, gabbay:nomr-jv}.

\subsection{Terms and signatures}
\label{ssec.na.sig}

\begin{defn}
Fix disjoint countably infinite collections of \deffont{atoms},
\deffont{unknowns} (or variables), and \deffont{term-formers} (or
function symbols).  We write $\mathbb A$ for the set of atoms;
$a,b,c,\ldots$ will range over distinct atoms.  $X,Y,Z,\ldots$ will
range over distinct unknowns.  $\tf f,\tf g,\ldots$ will range over
distinct term-formers.  We assume that to each $\tf{f}$ is associated
an \deffont{arity} $n$ which is a nonnegative number; we write
${\tf{f}: n}$ to indicate that $\tf{f}$ has arity $n$.  A
\deffont{signature} $\Sigma$ is a set of term-formers with their
arities.
\end{defn}

\begin{defn}
\label{def.permutations}
A \deffont{permutation} $\pi$ is a bijection on atoms 
such that $\f{nontriv}(\pi)=\{a\mid \pi(a) \neq a\}$ is finite.

We write $(a\ b)$ for the \deffont{swapping} permutation that maps $a$
to $b$, $b$ to $a$ and all other $c$ to themselves, and $\Id$ for the
\deffont{identity permutation}, so ${\Id(a)=a}$.  The notation 
${\pi\circ\pi'}$ is used for \deffont{functional composition} of permutations,
so $(\pi\circ\pi')(a)=\pi(\pi'(a))$, and ${\pi^\mone}$ for
\deffont{inverse}, so $\pi(a)=b$ if and only if $a=\pi^\mone(b)$.
\end{defn}

\begin{defn}
\label{defn.terms}
\deffont{(Nominal) terms} are inductively defined by:
\begin{gather*}
s,t,l,r,u\ \ ::=\ \ a\ \mid\ \pi\cdot X\ \mid\ [a]t\ \mid\ \tf{f}(t_1,\ldots,t_n) 
\end{gather*}
We write $\equiv$ for syntactic identity, so $t\equiv u$ when $t$ and $u$ denote the same term.
\end{defn}
A term of the form $[a]t$ is called an \deffont{(atom-)abstraction}; it represents `$x.e$' or `$x.\phi$' 
in expressions like `${\lambda x.e}$' or `${\forall x.\phi}$'. 
We define an $\alpha$-equivalence relation $\aleq$ later, in Definition~\ref{def.aleq}.

\subsection{Permutation and substitution}

\begin{defn}
\label{def.calculus.actions.permutation.object}
An \deffont{(atoms) permutation action} ${\pi \cdot t}$ is defined by: 
$$
\begin{array}{r@{\ }l@{\qquad}r@{\ }l}
\pi \cdot a                       \equiv & \pi(a)
&
\pi \cdot (\pi' {\cdot} X)        \equiv & (\pi \circ \pi') \cdot X
\\
\pi \cdot [a]t                   \equiv & [\pi(a)] (\pi \cdot t)
&
\pi \cdot \tf{f}(t_1, \ldots,t_n) \equiv & \tf{f}(\pi \cdot t_1, \ldots,\pi \cdot t_n)
\end{array}
$$

A \deffont{substitution (on unknowns)} $\sigma$ is a partial function from unknowns to terms with finite domain. 
$\theta$ and $\sigma$ will range over substitutions.

\label{def.calculus.actions.substitution}
An \deffont{(unknowns) substitution action} $t\sigma$ is defined by:
$$
\begin{array}{r@{\ }l@{\qquad}r@{\ }l@{\quad}l}
a\sigma                   \equiv& a
&
(\pi\act X)\sigma \equiv & \pi\act X & (X\not\in\f{dom}(\sigma))
\\
([a]t)\sigma              \equiv& [a](t\sigma)
&
(\pi \cdot X)\sigma       \equiv& \pi \cdot \sigma(X) & (X\in\f{dom}(\sigma))
\\
\tf{f}(t_1, \ldots, t_n)\sigma \equiv& \tf{f}(t_1\sigma, \ldots,t_n\sigma)
\end{array}
$$

Henceforth, if $X\not\in\f{dom}(\sigma)$ then $\sigma(X)$ denotes $\id\act X$.

We write $\id$ for the substitution with $\f{dom}(\id)=\varnothing$, so
that $t\id\equiv t$. When we write $\id$, it will be
clear whether we mean `$\id$ the identity substitution' or `$\id$ the
identity permutation' (Definition~\ref{def.permutations}).

If $\sigma$ and $\theta$ are substitutions, $\sigma\circ\theta$ maps each $X$ to $(X\sigma)\theta$.
\end{defn}

Lemmas~\ref{lem.equiv.perm}, \ref{lem.comm} and \ref{lemm.compose.subs} are proved by routine inductions (see~\cite{ gabbay:nomr-jv}).
\begin{lemm}
\label{lem.equiv.perm}
$(\pi \circ \pi') \cdot t \equiv \pi \cdot (\pi' \cdot t)$ \enspace and \enspace $\Id \cdot t \equiv t$.
\end{lemm}

\begin{lemm}
\label{lem.comm}
$\pi \act (t\sigma) \equiv (\pi \act t)\sigma$.
\end{lemm}

\begin{lemm}
\label{lemm.compose.subs}
$ t(\sigma\circ\theta)\equiv(t\sigma)\theta$.
\end{lemm}

\begin{figure*}[t]
$$
\begin{array}{c@{\quad}c}
\begin{prooftree}
\phantom{a \# t}
\justifies \Delta\cent a\# b
\using \rulefont{\#ab}
\end{prooftree}
&
\begin{prooftree}
\phantom{a \# t}
\justifies \Delta\cent a\# [a]t
\using \rulefont{\#[a]}
\end{prooftree}
\\[4ex]
\begin{prooftree}
(\pi^\mone(a)\# X)\in\Delta \justifies \Delta\cent a\# \pi\cdot X
\using \rulefont{\#X}
\end{prooftree}
&
\begin{prooftree}
\Delta\cent a\# t 
\justifies 
\Delta\cent a\# [b]t
\using \rulefont{\#[b]}
\end{prooftree}
\\[4ex]
\begin{prooftree}
\Delta\cent a\# t_1\: \cdots \:  \Delta\cent a\# t_n 
\justifies 
\Delta\cent a\# \tf{f}(t_1, \ldots, t_n)
\using \rulefont{\#\tf{f}}
\end{prooftree}
&
\begin{prooftree}
\phantom{h}
\justifies
\Delta\cent a\aleq a
\using \rulefont{{\aleq}a}
\end{prooftree}
\\[4ex]
\begin{prooftree}
\Delta\cent b\#t \ \ \Delta\cent (b\ a)\act t \aleq u
\justifies \Delta\cent [a]t  \aleq  [b]u 
\using \rulefont{{\aleq}[b]}
\end{prooftree}
&
\begin{prooftree}
(a\#X\in\Delta\text{ for all $a$ s.t. $\pi(a)\neq \pi'(a)$})
\justifies
\Delta\cent \pi\act X \aleq \pi'\act X
\using \rulefont{{\aleq}X}
\end{prooftree}
\\[4ex]
\begin{prooftree}
\Delta\cent t  \aleq  u \justifies \Delta\cent [a]t  \aleq  [a]u 
\using \rulefont{{\aleq}[a]}
\end{prooftree}
&
\begin{prooftree}
\Delta\cent t_i  \aleq  u_i\quad (1\leq i\leq n) 
\justifies 
\Delta\cent \tf{f}(t_1,\ldots,t_n)  \aleq  \tf{f}(u_1,\ldots,u_n)
\using \rulefont{{\aleq}\tf f}
\end{prooftree}
\end{array}
$$
\caption{Freshness and $\alpha$-equality}
\label{fig.rules.equality}
\end{figure*}

\subsection{$\alpha$-equivalence}

The native notion of equality on nominal terms is $\alpha$-equivalence.
For comparison, that of first-order terms is syntactic identity, and that of higher-order terms is $\beta$- or possibly $\beta\eta$-equivalence.

\begin{defn}
\label{def.aleq}
\label{def.calculus.freshness}
A \deffont{freshness (constraint)} is a pair ${a\#t}$ of an atom $a$ and a term $t$.  
We call a freshness of the form ${a\#X}$ \deffont{primitive}, and  a finite set of primitive freshnesses a \deffont{freshness context}. 
$\Delta$, $\Gamma$ and $\nabla$ will range over freshness contexts.

We may drop set brackets and write $a\#t,b\#u$ for $\{a\#t,b\#u\}$.
Also, we may write $a\#t,u$ for $a\#t,a\#u$, and $a,b \# t$ for $a\# t, b \#t$.

A \deffont{freshness judgement} is a tuple $\Delta\cent a\#t$ of a freshness context and a freshness constraint.
An \deffont{$\alpha$-equivalence judgement} is a tuple $\Delta\cent s\aleq t$ of a freshness context and two terms.
The \deffont{derivable} freshness and $\alpha$-equivalence judgements are defined by the rules in Figure~\ref{fig.rules.equality}.
\end{defn}

\begin{defn}
The functions $\f{atms}(t)$ and $\f{unkn}(t)$ will be used to compute the set of atoms and unknowns in a term, respectively. They are defined by:
$$
\begin{array}{r@{\ =\ }l@{\qquad}r@{\ =\ }l}
\f{atms}(a)&\{a\}
&
\f{atms}(\pi\act X)&\f{nontriv}(\pi)
\\
\f{atms}([a]t)&\f{atms}(t)\cup\{a\}
&
\f{atms}(\tf f(t_1,\ldots,t_n))&\bigcup_i \f{atms}(t_i)
\\[1.5ex]
\f{unkn}(a)&\varnothing
&
\f{unkn}(\pi\act X)&\{X\}
\\
\f{unkn}([a]t)&\f{unkn}(t)
&
\f{unkn}(\tf f(t_1,\ldots,t_n))&\bigcup_i \f{unkn}(t_i)
\end{array}
$$
\end{defn}

\begin{defn}
Later in this paper, starting with Definition~\ref{defn.closed}, we find it useful to write $\f{atms}(\mathtt{X})$ and $\f{unkn}(\mathtt{X})$ for $\mathtt{X}$ something more complex than a term --- e.g. a list (as in `$\f{atms}(\Delta,s,t)$'), a term-in-context (as in `$\f{unkn}(\nabla\cent l)$'), or a substitution.
By this we mean the atoms or unknowns appearing anywhere within the brackets.
So $\f{atms}(\Delta,s,t)$ means $\{a\mid a\#X\in\Delta\text{ for some }X\}\cup\f{atms}(s)\cup\f{atms}(t)$.
Also, 
$
\f{atms}(\theta)=\bigcup\{\f{atms}(\theta(X))\mid X\in\f{dom}(\theta)\}.
$
\end{defn}

\begin{lemm}[Strengthening]
\label{lemm.strengthening}
Suppose $a \not\in \f{atms}(s,t)$. 
Then:
\begin{itemize*}
\item
$\Delta, a \#X \cent b \# s$ implies $\Delta\cent b\#s$. 
\item
$\Delta, a \#X\cent s \aleq t$ implies $\Delta\cent s\aleq t$.
\end{itemize*}
\end{lemm}
\begin{proof}
By induction on the rules in Figure~\ref{fig.rules.equality}, using the fact that in all cases the hypotheses of rules use only atoms already mentioned in the conclusions.
\end{proof}

\begin{defn}
\label{defn.Stheta}
Suppose $S$ is a set of freshness constraints and $\theta$ is a substitution.
Define $S\theta=\{a\#(s\theta)\mid a\#s\in S\}$.

\end{defn}

\begin{lemm}[Weakening]
\label{lemm.weakening}
Suppose $\Delta\cent \Delta'\sigma$.
Then
\begin{itemize*}
\item
$\Delta'\cent b\#s$ implies $\Delta\cent b\#s\sigma$.
\item
$\Delta'\cent s\aleq t$ implies $\Delta\cent s\sigma\aleq t\sigma$.
\end{itemize*}
In particular, taking $\sigma=\id$ and $\Delta'=\Delta,\Gamma$, we obtain:
\begin{itemize*}
\item
$\Delta\cent b\#s$ implies $\Delta,\Gamma\cent b\#s$.
\item
$\Delta\cent s\aleq t$ implies $\Delta,\Gamma\cent s\aleq t$.
\end{itemize*}
\end{lemm}
\begin{proof}
By routine inductions on the rules in Figure~\ref{fig.rules.equality}. 
\end{proof}

\section{Nominal algebra and nominal rewriting}
\label{ssec.na.theories}

In this section we define notions of equational reasoning and rewriting
over nominal terms.  Nominal terms have a native notion of binding,
which theories inherit and can exploit to axiomatise properties of
binding operators (e.g. it is direct and natural to axiomatise
$\beta$-equivalence \cite{gabbay:lamcna}).

\begin{defn}
\label{defn.judgements}
We introduce two new judgement forms:
\begin{itemize*}
\item
An \deffont{equality judgement} is a tuple ${\Delta\cent s=t}$ of a freshness context and two terms. 
\item
A \deffont{rewrite judgement} is a tuple ${\Delta\cent s\to  t}$ of a freshness context and two terms. 
\end{itemize*}
We may write `${\varnothing \cent}$' as `${\cent}$'.
\label{defn.theory}

We also introduce two notions of theory --- one for equality judgements, and one for rewrite judgements:
\begin{itemize*}
\item
An \deffont{equational theory} ${\theory{T} = (\Sigma, \f{Ax})}$ is a pair of a signature $\Sigma$ and a possibly infinite set of equality judgements $\f{Ax}$ in that signature; we call them \deffont{axioms}.
\item
A \deffont{rewrite theory} ${\theory{R} = (\Sigma, \f{Rw})}$ is a pair of a signature $\Sigma$ and a possibly infinite set of rewrite judgements $\f{Rw}$ in that signature; we call these \deffont{rewrite rules}.
\end{itemize*}
We may omit $\Sigma$, identifying $\theory{T}$ with  $\f{Ax}$ and  $\theory{R}$ with
$\f{Rw}$ when the signature is clear from the context.
\end{defn}

\begin{xmpl}
The rewrite rules $\rulefont{\beta_{\app}}$, $\rulefont{\beta_{\vr}}$, $\rulefont{\beta_\epsilon}$, $\rulefont{\beta_\fn}$, and $\rulefont{\eta}$ define the rewrite theory $\beta\eta$ for $\beta$- and $\eta$-reduction in the $\lambda$-calculus. 

Note the use of a freshness context in rule  $\rulefont{\beta_\fn}$ to ensure that free $\lambda$-calculus variables are not captured. 
In rule $\rulefont{\beta_\epsilon}$ we use a freshness context to discard the argument when it is not needed. In the $\eta$ rule, the freshness context formalises the usual condition on the bound variable. See~\cite{gabbay:nomr-jv} for more examples of nominal rewrite rules. 

If we replace $\to$ by $=$ we obtain an equational theory. More examples of nominal equational theories can be found in~\cite{gabbay:nomuae}.
\end{xmpl}

\begin{defn}
\label{defn.positions}
A \deffont{position} $C$ is a pair $(s,X)$ of a term and a distinguished unknown $X$ that occurs precisely once in $s$, as $\id\act X$.
If $C=(s,X)$ then we write $C[t]$ for $s[X\sm t]$. 
\end{defn}

We are now ready to define notions of derivable equality, and rewriting:

\begin{defn}
\label{rewrite-step}
Below we write $\Delta\cent(\phi_1,\ldots,\phi_n)$ for the judgements $\Delta\cent \phi_1$, \ldots, $\Delta\cent\phi_n$.
\begin{itemize*}
\item
\deffont{Nominal rewriting:} The \emph{one-step rewrite relation} $\Delta\cent s\rewriteswith{R} t$
is the least relation such that  for every $(\nabla\cent l\rewriteswith{} r)\in\theory R$, freshness context $\Delta$, position $C$, term $s'$, permutation $\pi$, and substitution $\theta$, 
\begin{equation}
\label{eq.nr}
\hspace{-1em}
\begin{prooftree}
s\equiv C[s']
\qquad \Delta\cent\bigl(\nabla\theta,\quad s' \aleq \pi\act (l\theta),\quad C[\pi\act (r\theta)]\aleq t\bigr)
\justifies
\Delta\centt{} s\rewriteswith{R} t 
\using\rulefont{Rew_{\nabla\cent l\rewriteswith{} r}}
\end{prooftree} .
\end{equation}

The \emph{rewrite relation} $\Delta\centt{R} s\rewriteswith{} t$ is the reflexive transitive closure of the one-step rewrite relation, that is, the least relation that includes the one-step rewrite relation and such that:
\begin{itemize*}
\item
for all $\Delta$ and $s$: $\Delta\centt{R} s\rewriteswith{} s'$ if $\Delta\cent s\aleq s'$ (the native notion of equality of nominal terms is $\alpha$-equality); 
\item
for all $\Delta$, $s$, $t$, $u$:  $\Delta\centt{R} s\rewriteswith{} t$ and $\Delta\centt{R} t\rewriteswith{} u$ implies
$\Delta\centt{R} s\rewriteswith{} u$.
\end{itemize*}
If $\Delta\centt{R} s\rewriteswith{} t$ holds, we say that $s$ rewrites to $t$  in the context $\Delta$.  
\item
\deffont{(Nominal algebra) equality:} $\Delta\centt{T} s=t$ is the least transitive reflexive symmetric relation such that for every $(\nabla\cent l=r)\in\theory T$, freshness context $\Delta$, position $C$, permutation $\pi$, substitution $\theta$, and fresh $\Gamma$ (so if $a\#X\in\Gamma$ then $a\not\in\f{atms}(\Delta,s,t)$), 
\begin{equation}
\label{eq.na}
\begin{prooftree}
\Delta,\Gamma\cent \bigl(\nabla\theta,\quad s\aleq C[\pi\act (l\theta)],\quad C[\pi\act (r\theta)]\aleq t\bigr)
\justifies
\Delta\centt{T} s=t 
\using\rulefont{Axi_{\nabla\cent l=r}}
\end{prooftree} .
\end{equation}
\end{itemize*}
\end{defn}

We illustrate \eqref{eq.nr} and \eqref{eq.na} with examples.
\begin{xmpl}
\begin{itemize*}
\item
Consider the theories $\cent [a]X\to X$ and $\cent [a]X=X$.  We can
show that $[b][a]a$ rewrites to $[a]b$ in the empty freshness context,
that is, $\centt{\cent [a]X \to X} [b][a]a \to [a]b$. For this, we first use reflexivity
to transform $[b][a]a$ into $[a][b]b$ and then apply the rewrite rule
at position $C=([a]X,X)$. We can also show $\centt{\cent [a]X=X}
[b][a]a = [a]b$.
\item
Consider the rewrite theory $\beta \eta$ defining $\beta$- and $\eta$-reduction in the $\lambda$-calculus (see the Introduction). We can show that $\centt{\beta\eta} \tf{app}(\tf{lam}([a]\tf{app}(a,a)),b) \to \tf{app}(b,b)$ using rules $\rulefont{\beta_{\app}}$ and $\rulefont{\beta_{\vr}}$. 
\end{itemize*}
\end{xmpl}

\subsection{Equivalence with the literature}

The notions of equality and
rewriting in \eqref{eq.nr} and \eqref{eq.na} correspond to those 
in~\cite{gabbay:nomuae} and \cite{gabbay:nomr-jv} respectively.
However, the presentation of \eqref{eq.nr} and \eqref{eq.na} is original to this paper.
Arguably, Definition~\ref{rewrite-step} contains the clearest
presentation of nominal rewriting and nominal algebra so far.  It is
certainly the most concise and it makes it easier
 to compare and constrast the two notions --- to bring
out what they have in common, and what is different.

Some checking needs to be done to verify that \eqref{eq.nr} and
\eqref{eq.na} coincide with  nominal rewriting and nominal
algebra as presented in the literature.  
 All the main issues are indicated in
the following two short sketches:

\begin{rmrk}[Nominal rewriting]
\label{rmrk.nomr}
\eqref{eq.nr} corresponds to Definition~47 in Subsection~5.2 of \cite{gabbay:nomr-jv}. 
The correspondence is clear except that Definition~47 does not include a $\pi$.
This is because in \cite{gabbay:nomr-jv} rewrite theories (Definition~\ref{defn.theory} in this paper) have the additional property that they be \emph{equivariant} (Definition~4.2 of \cite{gabbay:nomr-jv}). This means that if $R\in\theory R$ then $R^\pi\in\theory R$ ($R^\pi$ is $R$ with $\pi$ applied to all atoms).
It is not hard to use Lemma~41 and part~(3) of Theorem~50 in \cite{gabbay:nomr-jv} to demonstrate that equivariance has the same effect as the $\pi$ in \eqref{eq.nr}, and indeed, if $\Delta\centt{R} s\rewriteswith{} t$ then $\Delta\centt{R} \pi\act s\to\pi\act t$. 

\end{rmrk}

\begin{rmrk}[Nominal algebra]
\eqref{eq.na} corresponds to Definition~3.10 and to the rules in Figures~1 and~2 in \cite{gabbay:nomuae}. 
The $C$ corresponds to the congruence rules \rulefont{cong[]} and \rulefont{cong\tf f}; the $\pi$ corresponds to the $\pi$ in \rulefont{ax} (modulo the same issue with $r^\pi$ versus $\pi\act r$ mentioned in Remark~\ref{rmrk.nomr}); \rulefont{perm} is built into $\aleq$.

Figure~2 of \cite{gabbay:nomuae} has an extra rule \rulefont{fr}, which generates a fresh atom. 
This corresponds to the fresh context $\Gamma$ in \eqref{eq.na}.
However, in \eqref{eq.na} the fresh atoms are generated `all at once', whereas in Figure~2 of \cite{gabbay:nomuae} fresh atoms may be generated at any point during equality reasoning.

We inspect the rules in Figure~2 of \cite{gabbay:nomuae} and see that we can commute an instance of \rulefont{fr} down through the other derivation rules; \rulefont{fr} is a structural rule, which adds freshness assumptions and does not affect the subgoal to be proved. 

If extra atoms in the derivation `accidentally clash' with the atom generated by the instance of \rulefont{fr},
then we rename the fresh atom in the subderivation to be `even fresher'.
The reader familiar with the proof of weakening for first-order logic can recall how we rename the bound variable in the $\forall$-right rule to be fresh for the weakened context; the proof obligation here is identical and does not involve any `nominal' elaborations.

Formally, an explicit inductive argument or the principle of ZFA equivariance \cite{gabbay:newaas-jv} prove that fresh atoms do not matter up to renaming, so the renamed subderivation is still a valid subderivation.
The interested reader is also referred to Lemma~5.10 in \cite{gabbay:oneaah-jv} where a similar result is stated and proved in full detail, of a more complex system. 
\end{rmrk}

\section{Soundness and completeness of nominal rewriting with respect to nominal algebra}
\label{sec:compl}
Theorem~\ref{prop:RimpliesT} and Theorem~\ref{thrm.compl} describe how nominal rewriting  relates to nominal algebra.

\begin{defn}
\label{defn.presentation}
Suppose $\theory T$ is an equational theory and \theory R is a rewrite
theory. We say that $\theory R$ is a \deffont{presentation} of $\theory
T$ if 
$$
\nabla\cent s=t\in\theory T\quad\liff\quad
(\nabla\cent s\rewriteswith{} t\in\theory R \ \ \lor\ \ \nabla\cent t\rewriteswith{} s\in\theory R) .
$$ 

We write $\Delta\centt{R} s\tworewriteswith{} t$ for the symmetric 
closure $\Delta\centt{R} s\rewriteswith{} t$. 
\end{defn}

\begin{prop}[Soundness]
\label{prop:RimpliesT}
Suppose $\theory R$ is a presentation of $\theory T$. 

Then $\Delta\centt{R} s\tworewriteswith{} t$ implies $\Delta\centt{T} s=t$.
\end{prop}
\begin{proof}
By a routine induction on the derivation $\Delta\centt{R} s\tworewriteswith{} t$.
We briefly sketch the case of \rulefont{Rew_{\nabla\cent l\rewriteswith{} r}} for $\nabla\cent l=r\in\theory T$.

Suppose for some $C$, $\theta$, and $\pi$, 
$$
s\equiv C[s']\quad\text{and}\quad \Delta\cent (\nabla\theta,\quad s'\aleq \pi\act (l\theta),\quad C[\pi\act (r\theta)]\aleq t).
$$
Let $\Gamma=\varnothing$.
It is a fact that if $\Delta\cent s'\aleq \pi\act (l\theta)$ then $\Delta \cent C[s'] \aleq C[\pi\act (l\theta)]$.
We now easily obtain an instance of \rulefont{Axi_{\nabla\cent l=r}}.
\end{proof}

\begin{rmrk}
\label{lemm.not.so}
Suppose $\theory R$  is a presentation of $\theory T$. 
It is not necessarily the case that $\Delta\centt{T} s=t$ implies $\Delta\centt{R} s\tworewriteswith{} t$.
To see this, take $\theory T = \{a\#X\cent X =  \tf{f}(X)\}$ and 
$\theory R=\{a\#X\cent X\rewriteswith{} \tf{f}(X)\}$.
Then $\centt{T} X=\tf{f}(X)$ (using $\rulefont{Axi}$ with $\Gamma = a \#X$),
but $\not\centt{R} X\tworewriteswith{} \tf{f}(X)$.
\end{rmrk}
\begin{thrm}[Quasi-Completeness]
\label{thrm.compl}
Suppose $\theory R$  is a presentation of $\theory T$. 

Then $\Delta\centt{T} s=t$ implies that there exists some fresh $\Gamma$ (so if $a\#X\in\Gamma$ then $a\not\in\f{atms}(\Delta,s,t)$) such that $\Delta,\Gamma\centt{R} s\tworewriteswith{} t$.
\end{thrm}
Note the `fresh $\Gamma$' on the side of nominal rewriting.
\begin{proof}
We work by induction on the derivation of $\Delta\centt{T} s=t$, write it $\Pi$.

The interesting case is \rulefont{Axi_{\nabla\cent l=r}} for some $\nabla\cent l=r\in\theory T$, of course.
There is only one argument in the proof that is not obvious: 
$\Pi$ is finite, so let us consider \emph{all} the finitely many instances of \rulefont{Axi} in $\Pi$; write them $I_1$, \ldots, $I_n$.
For each $1\leq i\leq n$,\ $I_i$ will involve $\nabla_i\cent l_i=r_i$, $C_i$, $\pi_i$, $\theta_i$, and a context $\Gamma_i$.
(Note that $\Delta$ is constant across all these instances.) 

Atoms in $\Gamma_i$ do not feature in $\Delta$, $C_i$, $\pi_i$, and $\theta_i$ --- they are `locally fresh'.
However, they might `accidentally' feature elsewhere in $\Pi$. 
It is a fact that because the atoms in $\Gamma_i$ do not feature in $\Delta$, $C_i$, $\pi_i$, and $\theta_i$, they do not feature in the conclusion of $I_i$. 
Therefore, it is a fact that we can rename these atoms so that they are fresh for all parts of $\Pi$ other than hypotheses of instances of \rulefont{Axi}, that is,  there exists a derivation $\Pi'$ of $\Delta\centt{T} s=t$ such that for each $1\leq i\leq n$ the respective $\Gamma_i'$ in respective instances $I_i'$ of  \rulefont{Axi} are fresh not only locally for the conclusion of $I_i'$, but also fresh globally for all conclusions of all $I_j'$ for $1\leq j\leq n$ in $\Pi'$. 
This `global freshness' condition is clearly preserved by taking subderivations.
We now take $\Gamma=\bigcup_i \Gamma_i'$, and the proof is by a routine induction on $\Pi'$.
Thus, an upper bound on $\Gamma$ is the maximal size of the $\Gamma_i$.

Note that although  $\rulefont{Rew_{\nabla\cent l=r}}$  appears to be more restrictive than $\rulefont{Axi_{\nabla\cent l=r}}$  in that \rulefont{Rew} requires $s \equiv C[s']$ and $s'$ $\alpha$-equivalent to an instance of a left-hand side, this is not an issue because the rewrite relation is transitive and includes the $\alpha$-equivalence relation.
\end{proof}

\section{Closed rewriting and nominal algebra}
\label{sec:closedrewriting} 

Theorem~\ref{thrm.compl} establishes a completeness result for nominal
rewriting modulo additional freshness constraints (the extra
$\Gamma$). 

This mismatch between nominal rewriting and nominal algebra could be
solved by including fresh atom generation in the definition of a
rewriting step.  But this comes at a cost --- the freshness
context may change along a rewrite derivation, and with it also the notion of $\alpha$-equivalence --- and is not needed for
large classes of systems, as we show below.

In this
section, we show that closed nominal rewriting is complete for nominal
algebra equality when all the axioms are \emph{closed}.  

Although
there are interesting systems, such as the axiomatisation of the
$\pi$-calculus~\cite{FernandezM:nomrng,gabbay:nomr-jv}, which are not
closed, this result has many applications: all the systems that arise
from functional programming (including the axiomatisation of the
$\lambda$-calculus) are closed, and all the systems that can be
specified in a standard higher-order rewriting formalism are also
closed (see~\cite{FernandezM:nomrng}).

\subsection{The definition of closed rules and closed rewriting}

\begin{defn}[Terms-in-context and nominal matching]
\label{defn.matching}
A \deffont{term-in-context} is a pair $\Delta\cent s$ of a freshness context and a term.

A \deffont{nominal matching problem} is a pair of terms-in-context 
$$(\nabla\cent l)\qleq (\Delta \cent s)
\quad\text{where}\ \ 
\f{unkn}(\nabla\cent l)\cap\f{unkn}(\Delta\cent s)=\varnothing. 
$$
A \deffont{solution} to this problem is a substitution $\sigma$ such that
$$
\Delta\cent \nabla\sigma 
\quad\text{and}\quad
\Delta\cent l\sigma\aleq s
\quad\text{and}\ \ 
\f{dom}(\sigma)\subseteq \f{unkn}(\nabla\cent l) . 
$$
\end{defn}

\begin{rmrk}
Nominal matching is decidable~\cite{gabbay:nomu-jv},
and can be solved in linear time~\cite{FernandezM:wollic08}.
\end{rmrk}

\begin{defn}[Freshened variants]
If $t$ is a term, we say that $\nw{t}$ is a \deffont{freshened variant} of $t$ when $\nw{t}$ has the same structure as $t$, except that the atoms and unknowns have been replaced by `fresh' atoms and unknowns (so they are not in $\f{atms}(t)$ and $\f{unkn}(t)$, and perhaps are also fresh with respect to some atoms and unknowns from other syntax, which we will always specify).
We omit an inductive definition.

Similarly, if $\nabla$ is a freshness context then $\nw{\nabla}$ will denote a freshened variant of $\nabla$ (so if $a\#X\in\nabla$ then $\nw{a}\#\nw{X}\in\nw{\nabla}$, where $\nw{a}$ and $\nw{X}$ are chosen fresh for the atoms and unknowns appearing in $\nabla$).

We may extend this to other syntax, like equality and rewrite judgements.

Note that if $\nw{\nabla}\cent \nw{l}\to \nw{r}$ is a freshened variant of $\nabla\cent l\to r$ then $\f{unkn}(\nw{\nabla}\cent \nw{l}\to \nw{r})\cap\f{unkn}(\nabla\cent l\to r)=\varnothing$.
\end{defn}

\begin{xmpl}
For example:
\begin{itemize*}
\item
$[\nw{a}][\nw{b}]\nw{X}$ is a freshened variant of $[a][b]X$,\ $\nw{a}\#\nw{X}$ is a freshened variant of $a\#X$,\ and $\varnothing\cent \nw{a}\rewriteswith{} \nw{b}$ is a freshened variant of $\varnothing\cent a\rewriteswith{} b$.
\item
Neither $[\nw{a}][\nw{a}]\nw{X}$ nor $[\nw{a}][\nw{b}]X$ are freshened variants of $[a][b]X$: the first, because we have wrongly identified two distinct atoms when we freshened them; the second, because we did not freshen $X$.
\end{itemize*}
\end{xmpl}

\begin{defn} 
\label{defn.closed}
A term-in-context $\nabla\cent l$ is \deffont{closed} if there exists a solution for 
the matching problem 
\begin{equation}
\label{eq.matching.problem}
(\nw{\nabla}\cent \nw{l})\ \ \qleq\ \ (\nabla ,\f{atms}(\nw{\nabla},\nw{l})\#\f{unkn}(\nabla,l)\cent l) . 
\end{equation} 
\end{defn}

\begin{lemm}
\label{lemm.unpack.closed.matching}
$\nabla\cent l$ is closed when there exists a substitution $\sigma$ with $\f{dom}(\sigma)\subseteq \f{unkn}(\nw{\nabla}\cent \nw{l})$ such that
$
\nabla,\f{atms}(\nw{\nabla},\nw{l})\#\f{unkn}(\nabla,l) \cent (\nw{\nabla}\sigma,\ l\aleq \nw{l}\sigma). 
$
\end{lemm}

\begin{defn}
\label{defn.closed.rewrite}
\begin{itemize*}
\item
Call $R=(\nabla\cent l\rewriteswith{} r)$ and $A=(\nabla\cent l=r)$ \deffont{closed} when $\nabla\cent (l,r)$ is closed\footnote{Here we use pair as a term former and apply the definition above.}.
\item
Given a rewrite rule $R=(\nabla\cent l\rewriteswith{} r)$ and a term-in-context $\Delta\cent s$, write 
$\Delta\cent s\rewriteswithc{R}t$
when there is some $\nw{R}$ a freshened variant of $R$ (so fresh for $R$, $\Delta$, $s$, and $t$), position $C$ and substitution $\theta$
such that
\begin{equation}
\label{eq.cnr}
\hspace{-2em} s\equiv C[s']
\ \ \ \text{and} \ \ \  
\Delta,\f{atms}(\nw{R})\mathrel{\#} \f{unkn}(\Delta,s,t)\cent (\nw{\nabla}\theta,\ s'{\aleq} \nw{l}\theta,\ C[\nw{r}\theta]{\aleq} t).
\end{equation}
We call this (one-step) \deffont{closed rewriting}.  

The \emph{closed rewrite relation} $\Delta\centt{R} s\rewriteswithc{} t$ is the reflexive transitive closure as in Definition~\ref{rewrite-step}.
\end{itemize*}
\end{defn}

The choice of freshened variant of $\nabla\cent l$ in
Definition~\ref{defn.closed} does not matter.  Similarly for
closed rewriting in Definition~\ref{defn.closed.rewrite}.  This
is related to the some/any property of the $\new$-quantifier
\cite{gabbay:newaas-jv}, and to the principle of ZFA equivariance
described e.g. in \cite[Theorem~A.4]{gabbay:nomuae}.  One way to look
at Definitions~\ref{defn.closed} and~\ref{defn.closed.rewrite} is that
the atoms in $\nw{\nabla}\cent \nw{l}$ occupy a `separate namespace'.

\begin{rmrk}
Closed nominal terms and rewriting were introduced in \cite{gabbay:nomr}.
$\Delta\cent s\rewriteswithc{R} t$ when $s$ rewrites to $t$ using a
  version of $R$ where the atoms and unknowns are renamed to be fresh.
  Renaming unknowns to be fresh is standard in rewriting, where
  variables in a rewrite rule are assumed distinct from those of the
  terms to be rewritten.  What is special about closed rewriting is
  that it applies a similar renaming to the atoms.

So for example, $\cent a\rewriteswith{a\to b} b$ and $\cent c\rewriteswith{a\to b} d$, but $\not\cent a\rewriteswithc{a\to b} b$ and $\not\cent c\rewriteswith{a\to b}d$.

A rule $R$ is closed when, intuitively, it is equal to any freshened variant $\nw{R}$ up to a substitution. 
$a\to b$ is not closed; the rules in \cite{gabbay:nomr-jv} for $\lambda$-calculus $\beta$-reduction are closed; those for $\pi$-calculus reduction are not closed.
\end{rmrk}

Comparing Definition~\ref{defn.closed.rewrite} (closed rewriting) with Definition~\ref{rewrite-step} (rewriting) we see they are very similar.
However, there are two key differences:
\begin{itemize*}
\item
The $\pi$ in \eqref{eq.nr} in Definition~\ref{rewrite-step} is not there in \eqref{eq.cnr} in Definition~\ref{defn.closed.rewrite}.
This $\pi$ can be very expensive \cite{CheneyJ:comeu}, so removing it greatly increases the efficiency of calculating closed nominal rewrites. 
\item
Atoms cannot `interact by name' in a closed rewrite step, because they are renamed.
\end{itemize*}

\subsection{Properties of closed rewriting, and connection with nominal algebra}

First we will prove a strengthening property for closed rewriting, for which we need some preliminary lemmas.

\begin{defn}
We define the substitution $\sigma\circ\pi$ by: 
$$
\begin{array}{r@{\ }c@{\ }l@{\qquad}l}
(\sigma\circ\pi)(X) &=& \pi\act(\sigma(X)) 
&\text{if }X\in\f{dom}(\sigma)
\\
(\sigma\circ\pi)(X)&&\text{undefined}
&\text{otherwise}.
\end{array}
$$
\end{defn}

\begin{lemm}
\label{lemm.commute.disjoint}
If $\f{atms}(s)\cap \f{nontriv}(\pi)=\varnothing$ then $\pi\act (s\sigma) \equiv s(\sigma\circ\pi)$.
\end{lemm}

\begin{lemm}
\label{lemm.nothing.silly}
\begin{enumerate}
\item
Suppose $a\not\in\f{atms}(s',\nw{l})$.
Then if $\Delta\cent s'\aleq \nw{l}\sigma$ then there exists $\sigma'$ such that $\Delta\cent \sigma(X)\aleq \sigma'(X)$ and $a\not\in\f{atms}(\sigma'(X))$, for all $X\in\f{unkn}(\nw{l})$.
\item
Suppose $a\not\in\f{atms}(t,\nw{r},C)$.
Then if $\Delta\cent C[\nw{r}\sigma]\aleq t$ then there exists some $\sigma'$ such that $\Delta\cent \sigma(X)\aleq\sigma'(X)$ and $a\not\in\f{atms}(\sigma'(X))$, for all $X\in\f{unkn}(\nw{r})$.
\end{enumerate}
\end{lemm}
\begin{proof}
For the first part, we construct $\sigma'$ by an induction on the structure of $\nw{l}$.
We sketch one case:
\begin{itemize*}
\item
The case $\nw{l}\equiv \pi\act X$.
\quad
By assumption $\Delta\cent s'\aleq \pi\act \sigma(X)$, where $a\not\in\f{nontriv}(\pi)$.
We choose $\sigma'(X)\equiv\pi^\mone\act s'$. 
\end{itemize*}
For the second part we work by induction on the derivation of $\Delta\cent C[\nw{r}\sigma]\aleq t$, using the rules in Figure~\ref{fig.rules.equality} to break down $C$ until we reach the first case (note that $\aleq$ is symmetric).
\end{proof}

\begin{lemm}
\label{lemm.aleq.sigma}
If $\Delta\cent\sigma(X)\aleq\sigma'(X)$ for all $X\in\f{unkn}(t)$ then $\Delta\cent t\sigma\aleq t\sigma'$.
\end{lemm}

\begin{prop}[Strengthening for closed rewriting] 
\label{prop.add.or.remove}
Fix a context $\Delta$ and terms $s$ and $t$.
Suppose $\Gamma$ is fresh (so if $a\#X\in\Gamma$ then $a \not\in \f{atms}(s,t,\Delta)$).
Suppose $R=(\nabla\cent l\rewriteswith{} r)$ is a rewrite rule.
Then $\Delta, \Gamma \cent s \rewriteswithc{R} t$  if and only if $\Delta \cent s \rewriteswithc{R} t$. 
\end{prop}
\begin{proof}
Suppose $\Delta, \Gamma \cent s \rewriteswithc{R} t$.
Unpacking definitions, there is some freshened $\nw{R}$ (with respect to $s$, $t$, $\Delta,\Gamma$, and $R$), and some position $C$ and substitution $\sigma$ such that $\f{dom}(\sigma)\subseteq \f{unkn}(\nw{R})$ and
\begin{align*}
s&\equiv C[s']
&\Delta,\Gamma,\f{atms}(\nw{R})\mathrel{\#} \f{unkn}(\Delta,s,t)&\cent (\nw{\nabla}\sigma,\ s'\aleq \nw{l}\sigma,\ C[\nw{r}\sigma]\aleq t).
\intertext{Using Lemmas~\ref{lemm.nothing.silly} and~\ref{lemm.aleq.sigma} we may assume without loss of generality that $a\not\in\f{atms}(\sigma)$.
By elementary calculations on the atoms of terms and using Strengthening (Lemma~\ref{lemm.strengthening})
we deduce
}
s&\equiv C[s']
&
\Delta,\f{atms}(\nw{R})\mathrel{\#} \f{unkn}(\Delta,s,t)&\cent (\nw{\nabla}\sigma,\quad s'\aleq \nw{l}\sigma,\quad C[\nw{r}\sigma]\aleq t).
\end{align*}
That is, $\Delta\cent s\rewriteswithc{R} t$ as required.

Conversely, suppose $\Delta\cent s \rewriteswithc{R} u$.
We unpack definitions as before and use the Weakening Lemma \ref{lemm.weakening}.
\end{proof}

We now establish the relationship between nominal rewriting and closed
rewriting. The first result, Proposition~\ref{prop.r.implies.cr}
below, shows that when a rule is closed, nominal rewriting implies
closed rewriting (this result was first shown as part of
\cite[Theorem~70]{gabbay:nomr-jv}; we give a shorter proof here). The
second result, Proposition~\ref{prop.cr.implies.r} below relating a
closed rewriting step with a nominal rewrite step, is new and is the
key to obtain a completeness proof for closed rewriting with respect
to nominal algebra.

\begin{lemm}
\label{lemm.fresh.equivar}
$\Delta\cent a\#s$ if and only if $\Delta\cent \pi(a)\#\pi\act s$.
\end{lemm}

\begin{prop}
\label{prop.r.implies.cr}
If $R=(\nabla \cent  l{\rewriteswith{}} r)$ is closed then $\Delta\cent s\rewriteswith{R} t$ implies $\Delta\cent s\rewriteswithc{R} t$.
\end{prop}
\begin{proof}
Suppose $\Delta\cent s\rewriteswith{R} t$.
So there exist $\Delta$, $C$, $s'$, $\pi$, and $\theta$ such that   
$$
s\equiv C[s']
\quad\text{and}\quad
\Delta\cent\bigl(\nabla\theta,\quad s' \aleq \pi\act (l\theta),\quad C[\pi\act (r\theta)]\aleq t\bigr)
.
$$
Without loss of generality we assume $\f{unkn}(\theta(X))\subseteq\f{unkn}(\Delta,s,t)$ for every $X\in\f{dom}(\theta)$ (because we only `use' the part of $\theta$ that maps $l$ to $s$ and $r$ to $t$).
 
$\nabla\cent l\rewriteswith{} r$ is closed so by Lemma~\ref{lemm.unpack.closed.matching} there is a freshened variant $\nw{R}=(\nw{\nabla}\cent \nw{l}\rewriteswith{} \nw{r})$ of $R$ and a substitution $\sigma$ such that $\f{dom}(\sigma)\subseteq \f{unkn}(\nw{R})$ and
\begin{align*}
\nabla,\f{atms}(\nw{l})\mathrel{\#}\f{unkn}(\Delta,s,t) &\cent (\nw{\nabla}\sigma,\quad l\aleq \nw{l}\sigma,\quad r\aleq \nw{r}\sigma) . 
\intertext{It is not hard to use our assumptions to verify that}
\Delta,\f{atms}(\nw{l})\mathrel{\#}\f{unkn}(\Delta,s,t) &\cent \f{atms}(\nw{l})\mathrel{\#}\f{unkn}(\Delta,s,t)\theta  .
\end{align*}
It follows using Lemmas~\ref{lemm.weakening} and~\ref{lem.comm}
that
\begin{align*}
s&\equiv C[s']
&\Delta,\f{atms}(\nw{l})\mathrel{\#}\f{unkn}(\Delta,s,t)&\cent\bigl(\nw{\nabla}\sigma\theta,\ s' {\aleq} \pi{\act} (\nw{l}\sigma\theta),\  C[\pi\act (\nw{r}\sigma\theta)]{\aleq} t\bigr)
.
\intertext{
By assumption the atoms in $\nw{R}$ are fresh and so we can assume 
$\f{atms}(\nw{R})\cap\f{nontriv}(\pi)=\varnothing$.
It follows by Lemmas~\ref{lemm.commute.disjoint} and~\ref{lemm.compose.subs} that
$\pi\act(\nw{l}\sigma\theta)\equiv \nw{l}((\sigma\circ\theta)\circ\pi)$
and 
$\pi\act(\nw{r}\sigma\theta)\equiv \nw{r}((\sigma\circ\theta)\circ\pi)$.
Using Lemma~\ref{lemm.fresh.equivar} $\Delta,\f{atms}(\nw{l})\#\f{unkn}(\nabla,l)\cent \nw{\nabla}((\sigma\circ\theta)\circ\pi)$ also follows. 
Write $\theta'$ for $(\sigma\circ\theta)\circ\pi$.
Then
}
s&\equiv C[s']
&\Delta,\f{atms}(\nw{l})\mathrel{\#}\f{unkn}(\Delta,s,t)&\cent\bigl(\nw{\nabla}\theta',\quad s' \aleq \nw{l}\theta',\quad C[\nw{r}\theta']\aleq t\bigr)
.
\end{align*}
That is, $\Delta\cent s\rewriteswithc{R} t$ as required.
\end{proof}

\begin{lemm}
\label{lemm.sigmainverse}
Suppose $\nabla\cent l$ is a closed term-in-context where $\f{atms}(\nabla\cent l)=\{a_1,\ldots,a_n\}$ and $\f{unkn}(\nabla\cent l)=\{X_1,\ldots,X_n\}$;  we take these atoms and unknowns in some fixed but arbitrary order.

Suppose is $\nw{\nabla}\cent \nw{l}$ a freshened variant of $\nabla\cent l$ where  $\f{atms}(\nw{\nabla}\cent \nw{l})=\{\nw{a}_1,\ldots,\nw{a}_n\}$ and $\f{unkn}(\nw{\nabla}\cent \nw{l})=\{\nw{X}_1,\ldots,\nw{X}_n\}$; we take these fresh atoms and unknowns in a corresponding order.  

Let $\tau$ and $\varsigma$ be the permutation and substitution defined by
$$
\tau = (\nw{a}_1\ a_1)\circ\ldots\circ(\nw{a}_n\ a_n) \quad\text{and}\quad
\varsigma = [X_1\ssm \tau\act \nw{X}_1,\ldots,X_n\ssm \tau\act \nw{X}_n],
$$  
then: 
\begin{enumerate*}
\item
$\nw{l}\equiv \tau\act(l\varsigma)$.
\item
$\Gamma'\cent \nw{\nabla}\theta$ if and only if $\Gamma'\cent \nabla\varsigma\theta$.
\end{enumerate*}
\end{lemm}
\begin{proof}
We prove the first part by induction on $l$.  
We sketch the case of $\pi\act X$:
$$
\tau\act ((\pi\act X)\varsigma) 
\stackrel{\text{Lemma~\ref{lem.equiv.perm}}}{\equiv} 
(\tau\circ\pi) \act\varsigma(X) 
\stackrel{\text{Lemma~\ref{lem.equiv.perm}}}{\equiv} 
(\tau\circ\pi\circ \varsigma) \act X 
\stackrel{\text{fact}}{\equiv} 
\pi'\act \nw{X}.
$$
For the second part consider some $\nw{a}\#\nw{X}\in\nw{\nabla}$ (originating from $a\#X\in\nabla$).
By definition $\varsigma(X)\equiv \tau\act \nw{X}$ and it follows that
$$
\nw{X}\theta
\stackrel{\text{Lemma~\ref{lem.equiv.perm}}}{\equiv}
(\tau\act (X\varsigma))\theta 
\stackrel{\text{Lemma~\ref{lem.comm}}}{\equiv}
\tau\act (X\varsigma\theta).
$$
By Lemma~\ref{lemm.fresh.equivar} $\Gamma'\cent \nw{a}\#(\nw{X}\theta)$ if and only if $\Gamma'\cent a\#(X\varsigma\theta)$.
The result follows.
\end{proof}

\begin{prop}
\label{prop.cr.implies.r}
If $R=(\nabla\cent l\rewriteswith{} r)$ is closed then $\Delta\cent s\rewriteswithc{R} t$ implies there is some fresh $\Gamma$ (so if $a\#X\in\Gamma$ then $a\not\in\f{atms}(\Delta,s,t)$) such that $\Delta,\Gamma\cent s\rewriteswith{R} t$.
\end{prop}
\begin{proof}
If $\Delta\centt{R} s\rewriteswithc{} t$ then for some freshened variant $\nw{R}=(\nw{\nabla}\cent \nw{l}\rewriteswith{} \nw{r})$ of $R$ (freshened with respect to $R$, $\Delta$, $s$, and $t$) there exists some position $C$, term $s'$, and substitution $\theta$ such that 
\begin{align*}
s&\equiv C[s'] &\Delta,\f{atms}(\nw{R})\mathrel{\#}\f{unkn}(\Delta,s,t)&\cent (\nw{\nabla}\theta,\quad s'\aleq \nw{l}\theta,\quad C[\nw{r}\theta]\aleq t) . 
\intertext{By Lemmas~\ref{lemm.sigmainverse} and~\ref{lem.comm}, there exists $\tau$ and $\varsigma$ such that:}
s&\equiv C[s'] & \Delta,\f{atms}(\nw{R})\mathrel{\#}\f{unkn}(\Delta,s,t)&\cent (\nabla\varsigma\theta,\ s'\aleq \tau\act (l\varsigma\theta),\ C[\tau\act (r\varsigma\theta)]\aleq t) . 
\end{align*}
Using Lemmas~\ref{lemm.compose.subs} and~\ref{lem.comm} we deduce
$
\Delta,\f{atms}(\nw{R})\mathrel{\#}\f{unkn}(\Delta,s,t)\cent s\rewriteswith{R} t.
$
\end{proof}
 
\begin{defn}
$\Delta\centt{R} s\tworewriteswithc{}t$ denotes the symmetric closure 
of $\Delta\centt{R} s\rewriteswithc{} t$.
\end{defn}

\begin{thrm}[Soundness and completeness]
\label{thm:sound-compl}
Suppose the rewrite theory $\theory R$ is a presentation (Definition~\ref{defn.presentation})
of the equational theory  $\theory T$.
Suppose all rules in $\theory R$ are closed.
Then $\Delta\centt{T} s=t$ if and only if $\Delta\centt{R} s\tworewriteswithc{}t$.
\end{thrm}
\begin{proof}
Suppose $\Delta\centt{T} s=t$.
By Theorem~\ref{thrm.compl} there is a fresh $\Gamma$ such that $\Delta,\Gamma\centt{R} s\tworewriteswith{} t$. 
By Proposition~\ref{prop.r.implies.cr} and Strengthening (Proposition~\ref{prop.add.or.remove}) $\Delta\centt{R} s\tworewriteswithc{} t$.

Conversely, suppose $\Delta\centt{R} s\tworewriteswithc{} t$.
By Proposition~\ref{prop.cr.implies.r} $\Delta,\Gamma\centt{R} s\tworewriteswith{} t$ for some fresh $\Gamma$.
It follows by Proposition~\ref{prop:RimpliesT} that $\Delta\centt{T} s=t$.
\end{proof}

\subsection{Mechanising equational reasoning}
\label{subsec:mechanising}

Closed nominal rewriting can be used 
to automate reasoning in nominal equational theories,
provided that the theory satisfies certain conditions.

\begin{defn}
A rewrite theory $\theory R$ is \deffont{closed} when every $R\in\theory R$ is closed (Definition~\ref{defn.closed.rewrite}). 
We say that $t$ is an \deffont{($\theory R$-)normal form} of $s$ if $\Delta \centt{R} s \rewriteswithc{} t$ and there is no $u$ such that $\Delta \centt{R} t \rewriteswithc{} u$ (so there is no rewrite from $t$).

A theory $\theory R$ is \deffont{terminating} when there are no infinite closed rewriting sequences $\Delta \centt{R} t_1  \rewriteswithc{} t_2,\ t_2  \rewriteswithc{} t_3,\ \ldots $.
It is  \deffont{confluent} when, if $\Delta\centt{R} s  \rewriteswithc{}
 t$ and $\Delta\centt{R} s \rewriteswithc{} t'$, then $u$ exists such that 
$\Delta\centt{R} t \rewriteswithc{} u$ and $\Delta\centt{R} t' \rewriteswithc{} u$.

A theory $\theory R$ is \deffont{convergent} when it is terminating and confluent.
\end{defn}

\begin{thrm}
\label{thrm.eq.rwc}
Suppose the axioms in a theory $\theory T$ can be
oriented to form a closed $\theory R$. If
$\theory R$ is confluent, then $\Delta \centt{T} s = t$ if and only if
there exists $u$ such that $\Delta \centt{R} s  \rewriteswithc{} u$ and $\Delta
\centt{R} t  \rewriteswithc{} u$.
\end{thrm}
\begin{proof}
By Theorem~\ref{thm:sound-compl}. 
\end{proof}

Theorem~\ref{thrm.eq.rwc} does not require termination.  If we have
termination then we can decide whether there exists a term $u$ with
the desired property: it suffices to rewrite $s$ and $t$ to normal
form and then check that the normal forms are $\alpha$-equivalent
(convergence guarantees existence and unicity of normal forms up to
$\alpha$-equivalence; a
linear-time algorithm to check $\alpha$-equivalence is described in~\cite{FernandezM:wollic-jv}). Also, since Theorem~\ref{thrm.eq.rwc} uses closed rewriting, the computation of a rewrite step is efficient: nominal matching is sufficient (see also~\cite{FernandezM:wollic-jv} for linear-time nominal matching algorithms).

\begin{corr}[Decidability of deduction in {\theory T}]
\label{cor:dec}
Suppose $\theory T$ is an equational theory whose axioms can be
oriented to form a closed $\theory R$. 
Suppose $\theory R$ is convergent. Then equality is decidable in $\theory T$
(i.e., $\Delta \centt{T} s = t$ is a decidable relation). 
\end{corr}

\section{Conclusions}
\label{sec:Conclusions}

Efficient algorithms for closed nominal rewriting and for checking
$\alpha$-equivalence are described in~\cite{FernandezM:wollic-jv}. We
can also check that rules are closed in linear time, with the nominal
matching algorithm of~\cite{FernandezM:wollic-jv}.  It follows from
Corollary~\ref{cor:dec} that, had we a procedure to check that a given
set of rules is convergent, we could directly build an automated
theorem prover for nominal theories. Unfortunately, termination and
confluence are undecidable properties even for first order
rules. Fortunately, closed nominal rewrite rules inherit many of the
good properties of first-order rewriting systems: orthogonality is a
sufficient condition for confluence (see~\cite{gabbay:nomr}) and it is
easy to check. If the theory under consideration is not
orthogonal, then the alternative is to check termination and to check
that all critical pairs are joinable (which is a sufficient condition
for convergence, see~\cite{gabbay:nomr}). Reduction orderings (to
check termination) and completion procedures (to ensure that all
critical pairs are joinable) are available for closed nominal
rules~\cite{FernandezM:nrpo}.

We can consider a recent `\emph{permissive}' variant of nominal terms \cite{gabbay:perntu,gabbay:unialt}.
These eliminate freshness contexts and give a tighter treatment of $\alpha$-equivalence, which might simplify the proofs here. 
Permissive nominal terms have been implemented in prototype form \cite{mulligan:imppnt}, but it remains to consider more efficient algorithms to manipulate them.

\newcommand{\etalchar}[1]{$^{#1}$}

\end{document}